\documentclass[lettersize,journal]{IEEEtran}

\usepackage{amsmath,amssymb}
\usepackage{subfigure}
\usepackage{hyperref}
\usepackage{graphicx,graphics,color,psfrag}
\usepackage{cite,balance}
\usepackage{caption}
\allowdisplaybreaks
\usepackage{algorithm}
\usepackage{algorithmic}
\usepackage{accents}
\usepackage{amsthm}
\usepackage{bm}
\usepackage{url}
\usepackage[english]{babel}
\usepackage{multirow}
\usepackage{enumerate}
\usepackage{cases}
\usepackage{stfloats}
\usepackage{dsfont}
\usepackage{color,soul}
\usepackage{amsfonts}
\usepackage{cite,graphicx,amsmath,amssymb}
\usepackage{subfigure}
\usepackage{fancyhdr}
\usepackage{hhline}
\usepackage{graphicx,graphics}
\usepackage{array,color}
\usepackage{mathtools}
\usepackage{amsmath}
\usepackage{amsthm}
\usepackage{caption}
\usepackage{circledsteps}

\newtheorem{remark}{Remark}
\newtheorem{corollary}{Corollary}
\newtheorem{proposition}{Proposition}

\usepackage{algorithm}
\def\BibTeX{{\rm B\kern-.05em{\sc i\kern-.025em b}\kern-.08em
		T\kern-.1667em\lower.7ex\hbox{E}\kern-.125emX}}
\begin{document}

	\title{Asynchronous Cell-Free Massive MIMO-OFDM: Mixed Coherent and Non-Coherent Transmissions }
	
		\author{\IEEEauthorblockN{Guoyu Li,~\IEEEmembership{Student Member,~IEEE},~Shaochuan Wu,~\IEEEmembership{Senior Member, IEEE},\\Changsheng You,~\IEEEmembership{Member,~IEEE},~Wenbin Zhang,~\IEEEmembership{Member, IEEE}, and~Guanyu Shang}

		\IEEEcompsocitemizethanks{\IEEEcompsocthanksitem This work was supported by the National Natural Science Foundation of China under Grant 62271167. (\emph{Corresponding author: Shaochuan Wu.}) } 
		\IEEEcompsocitemizethanks{\IEEEcompsocthanksitem G. Li, S. Wu and W. Zhang are with the School of Electronics and Information Engineering, Harbin Institute of Technology, Harbin 150080, China (email: lgy@stu.hit.edu.cn; scwu@hit.edu.cn; zwbgxy1973@hit.edu.cn). 
		C. You is with the Department of Electronic and Electrical  Engineering, Southern University of Science and Technology (SUSTech), Shenzhen 518055, China (e-mail: youcs@sustech.edu.cn).
		G. Shang is with the Innovation Photonics and Imaging Center, School of Instrumentation Science and Engineering, Harbin Institute of Technology, Harbin, China (email: shangguanyu95@163.com). } 
		
		
	}
	

	\maketitle
	
	\begin{abstract}
	
	In this letter, we analyze the performance of asynchronous cell-free multiple-input multiple-output orthogonal frequency division multiplexing (CF mMIMO-OFDM) system with mixed coherent and non-coherent transmission approaches.
	To this end, we first obtain the achievable downlink sum-rate for the mixed coherent and non-coherent transmissions, and then provide a closed-form expression for the case with the maximum ratio precoding. Subsequently, an efficient access point (AP) clustering algorithm is proposed to group APs into a set of coherent clusters. Numerical results demonstrate that the mixed coherent and non-coherent transmissions can effectively improve the sum-rate of CF mMIMO-OFDM systems under asynchronous reception.

	\end{abstract}
	

	\section{Introduction}

	As one of the key candidate technologies for the sixth-generation wireless networks, cell-free (CF) massive multiple-input multiple-output (mMIMO) can effectively eliminate cell boundaries and provide users with higher spectral efficiency~\cite{10422885}. In CF network architecture, a large number of access points (APs) coherently transmit the data stream to each user, achieving high data rates~\cite{7827017}. The advantage of coherent transmission suggests that even if the APs have different channel gains to the users, distributing the transmit power across multiple APs is more beneficial than transmitting only from the AP with the strongest channel~\cite{8809413}. 

	Although jointly coherent transmission generally provides higher transmission rates, it imposes strict requirements on the availability of perfect synchronization in CF mMIMO networks~\cite{10684238}, which, however, is practically challenging.  Specifically, due to different locations of APs and users in CF mMIMO networks, there inevitably exists time differences in data arrivals. 
	This asynchronous reception situation not only hinders the acquisition of accurate channel state information (CSI) but also introduces quantized phase shifts, resulting in coherent transmission impossible~\cite{10032129, 10319677}.
	For example, the authors in~\cite{10319677} shown that the quantized phase shift introduced by asynchronous reception can seriously reduce the CF mMIMO orthogonal frequency division multiplexing (OFDM) system achievable rate.
	Therefore, it is necessary to find efficient solutions to operate CF mMIMO-OFDM networks in the presence of asynchronous reception.

	Recently, the research~\cite{10279001, 10459246} has shed light on a viable solution. 
	The mixed coherent and non-coherent transmissions proposed in~\cite{10279001} can effectively cope with the phase misalignment problem in multi-central processing unit (CPU) CF mMIMO systems. Specifically, the authors assumed that APs connected to the same CPU are phase-aligned, which sends data coherently. Meanwhile, non-coherent transmission performs between APs that connect to different CPUs. 
	Moreover, the authors considered the user-centric approach and proposed extensions of existing clustering algorithms for multi-CPU CF mMIMO systems. 
	However, the above assumption does not take into account the practical phase misalignment problem, that is, APs connected to the same CPU may also have different phases. 
	To this end, the authors in~\cite{10459246} proposed an AP clustering algorithm to group the APs into phase-aligned clusters, solving the practical problem of phase misalignment. However, the above AP clustering schemes do not solve the asynchronous reception problem considered in~\cite{10032129} and~\cite{10319677}. The reason is that in asynchronous reception, there are generally different quantized phase shifts between the same AP and each user, resulting in different coherent clusters for each user.
	Moreover, the performance analysis of mixed coherent and non-coherent transmissions in addressing the impact of asynchronous reception on CF mMIMO-OFDM is lacking in the existing literature.
	
	
	Motivated by the observations discussed, we analyze the performance of mixed coherent and non-coherent transmission methods in mitigating the impact of asynchronous reception on CF mMIMO-OFDM systems. To this end, we first analyze the achievable downlink sum-rate for mixed coherent and non-coherent transmissions under asynchronous reception, and then obtain a closed-form expression for the case with the maximum ratio (MR) precoding. Furthermore, we propose an efficient AP clustering algorithm to group APs into a set of coherent clusters. Numerical results verify the effectiveness of the proposed methods. 
	


	\section{System Model}

	We consider a CF mMIMO-OFDM system, where $Q$ $M$-antenna APs and $K$ single-antenna users are randomly located in a large geographical area. 
	Each AP is connected to a CPU, via a wire or wireless fronthaul. 
	For each OFDM symbol, the sample length is $N_{\mathrm{OFDM}}= N+ N_{\mathrm{CP}}$, where $N$ is the number of subcarriers, and $N_{\mathrm{CP}}$ is the cyclic prefix (CP) length in samples. The sampling period and subcarrier space of the OFDM symbol are $T_s$ and $\Delta f$, respectively.
	
	\subsection{Channel Models}
	
	\subsubsection{Block-fading Channel Model}

	We assume the block-fading model throughout this letter, where the channel fading is time-invariant and frequency-flat in each coherence block (CB). The entire time-frequency resource is divided into $R_{\mathrm{CB}} = T_{\mathrm{CB}} N_{\mathrm{CB}}$ CBs, where $T_{\mathrm{CB}}$ and $N_{\mathrm{CB}}$ are the number of CBs in time and frequency, respectively. For an arbitrary CB $r \in \left\{ 1, 2, \cdots, R_{\mathrm{CB}} \right\}$, there contains $N_{\mathrm{sub}} = N / N_{\mathrm{CB}}$ consecutive subcarriers and $N_{\mathrm{T}}$ consecutive OFDM symbols. Without loss of generality, we conduct the performance analysis by studying a single statistically representative CB $\mathfrak{B}_{r}$. In CB $\mathfrak{B}_{r}$, we assume that the index of the first subcarrier is $n_{1}$. 
	At this time, the channel response vectors between user $k$ and AP $q$ in CB $\mathfrak{B}_{r}$ are identical i.e., $\mathbf{h}_{qk,n_{1}} = \cdots = \mathbf{h}_{qk,n_{1}+N_{\mathrm{sub}}-1} = \mathbf{h}_{qk}^{\mathfrak{B}_{r}}$, which is modeled by correlated Rayleigh fading as
	\begingroup\makeatletter\def\f@size{10}\check@mathfonts
		\begin{equation}
			\begin{aligned} \label{Channel_model}
				\mathbf{h}_{qk}^{\mathfrak{B}_{r}} \sim \mathcal{N}_{\mathbb{C}} \left( \mathbf{0}_{M},\mathbf{R}_{qk} \right),
			\end{aligned} 
		\end{equation}
	\endgroup 
	where $\mathbf{R}_{qk} \in \mathbb{C}^{M \times M}$ represents the spatial correlation matrix of the channel vector $\mathbf{h}_{q k}^{\mathfrak{B}_{b}^{r}}$, and $\beta_{qk} = \mathrm{tr} \left( \mathbf{R}_{qk} \right)/M$ is the large-scale fading coefficient.
	
	\subsubsection{Asynchronous Reception Channel Model}  

	
	Consider a user-centric CF mMIMO-OFDM downlink transmission scenario, where each user is served only by a subset of nearby APs. We denote by $\mathcal{Q}_{k}$ the set of APs that serve the $k$th user. Correspondingly, according to the set $\mathcal{Q}_{k}$  for $k=1, 2, \cdots, K$, the subset of users that are served by AP $q$ is denoted as $\mathcal{K}_{q}$.
	As the $Q$ APs are randomly distributed in the network, the signal propagation time durations from different APs to an arbitrary user are generally different. 
	Without loss of generality, we assume that the arrival time of the first received signal at user $k$ is $t_{\hat{q}_{k} k} = 0$ (AP $\hat{q}_{k}$ is the nearest AP to user $k$), and denote the quantized time offset from the $q$th AP to the $k$th user in the sampling interval as~\cite{8341954} 
	\begingroup\makeatletter\def\f@size{10}\check@mathfonts
		\begin{equation}
			\begin{aligned} \label{Quantized_propagation_delay}
				\delta_{q k} = \left \lfloor (d_{q k} - d_{\hat{q}_{k} k}) / (c \cdot T_{s}) \right \rfloor = \left \lfloor (t_{q k} - t_{\hat{q}_{k} k}) / T_{s} \right \rfloor, 
			\end{aligned} 
		\end{equation}
	\endgroup 
    where $\left \lfloor \cdot \right \rfloor$ is the floor function, $c$ is the speed of light, $d_{q k}$ ($t_{q k}$) and $d_{\hat{q} k}$ ($t_{\hat{q} k}$) are the propagation distance (time) from AP $q$ and AP $\hat{q}_{k}$ to user $k$, respectively.
	The received signal of user $k$ on subcarrier $n \in \left\{ n_{1}, n_{1}+1, \cdots, n_{1}+N_{\mathrm{sub}}-1 \right\}$ under asynchronous reception is given by
	\begingroup\makeatletter\def\f@size{10}\check@mathfonts
		\begin{equation}
			\begin{aligned} \label{System_model}
				y_{k, n} =& \sum\nolimits_{q \in \mathcal{Q}_{k}} (\chi_{q k,n} \mathbf{h}_{qk}^{\mathfrak{B}_{r}})^{H} \mathbf{w}_{q k} s_{k, n} \\
				&+ \sum\nolimits_{i \neq k} \sum\nolimits_{q \in \mathcal{Q}_{i}} (\chi_{q k,n} \mathbf{h}_{q k}^{\mathfrak{B}_{r}})^{H} \mathbf{w}_{q i} s_{i, n} + n_{k, n},
			\end{aligned} 
		\end{equation}
	\endgroup 
	where $\chi_{q k, n} = e^{-j 2 \pi n \delta_{q k} / N}$ represents the quantized phase shift on the $n$th subcarrier caused by the quantized time offset $\delta_{q k}$, $\mathbf{w}_{q k, n} \in \mathbb{C}^{M \times 1}$ is the precoding vector for user $k$, $s_{k, n}$ is the data signal for user $k$ with $\mathbb{E} \{ s_{k, n} s_{k, n}^{H} \} = 1$, $n_{k, n} \sim \mathcal{N}_{\mathbb{C}} \left( 0,\sigma^2 \right)$ is the additive white Gaussian noise at the $k$th user.
	\begin{remark} 

		Note that the quantized time offset in~\eqref{Quantized_propagation_delay} are all integers. According to~\eqref{Quantized_propagation_delay}, the quantized time offset from AP $\hat{q}_{k}$ to user $k$ is 0. Moreover, for any AP $q$, if the distance difference satisfies the condition $\delta_{q k} \cdot D \le  | d_{q k} - d_{\hat{q}_{k}} | < (\delta_{q k}+1) \cdot D$ ($D = T_{s} \cdot c$ is defined as the sampling distance), then the quantized time offset from AP $q$ to user $k$ is $\delta_{q k}$. As such, the maximum quantized time offset from AP in set $\mathcal{Q}_{k}$ to user $k$ is $\delta_{\check{q}_{k} k} = \left \lfloor (d_{\check{q}_{k} k} - d_{\hat{q}_{k} k}) / D \right \rfloor$, where $\check{q}_{k}$ is the farthest AP to user $k$.
		 
	\end{remark}

	\subsection{Uplink Channel Estimation}

    We assume that the system operates in the time-division duplex mode, where the CSI obtained from uplink channel estimation is used for downlink precoding. There are $\tau_p$ channel uses for uplink pilots, and $\tau_d = N_{\mathrm{sub}}N_{T}  - \tau_p $ channel uses for downlink data transmission. In the uplink training phase, a set of $\tau_p$-length mutually orthogonal frequency-multiplexing pilot sequences $\boldsymbol{\Phi} = \left\{ \boldsymbol{\phi}_{1}, \boldsymbol{\phi}_{2}, \cdots, \boldsymbol{\phi}_{\tau_{P}} \right\}$ with $\left\| \boldsymbol{\phi}_{k} \right\|^{2} = \tau_{p}$ and $\boldsymbol{\phi}_{i}^{H} \boldsymbol{\phi}_{k} = 0$ for $k \ne i$ is used to estimate the channel. Assume that the first OFDM symbol in CB $\mathfrak{B}_{r}$ is used to send the pilot sequence and $\tau_{p}=N_{\mathrm{sub}}$. The asynchronous received signal at AP $q$ can be expressed as
	\begingroup\makeatletter\def\f@size{10}\check@mathfonts
		\begin{equation} \label{channel_estimation}
			\mathbf{Y}_{q} = \sum\nolimits_{i=1}^{K} \sqrt{p_{i}} \mathbf{h}_{qi}^{\mathfrak{B}_{r}} \left( \mathbf{\Theta}_{q i}\boldsymbol{\phi}_{i} \right)^{T} +\mathbf{N}_{q}, 
		\end{equation}
	\endgroup
	where $\mathbf{\Theta}_{q i}=\mathrm{diag}(\chi_{q i, n_{1}}^{'}, \cdots, \chi_{q i, n_{1}+\tau_{p}-1}^{'}) \in \mathbb{C}^{\tau _p \times \tau _p}$ is the quantized phase shift diagonal matrix, $p_i$ is the transmit power of user $i$, $\mathbf{N}_{q} \in \mathbb{C}^{M \times \tau _p}$ is the receiver noise matrix with independent Gaussian entries following $\mathcal{N}_{\mathbb{C}}\left(0, \sigma^{2}\right)$~\cite{10319677}. Note that the quantized phase shift in the uplink and downlink are generally not equal because the quantized phase shift in the uplink is obtained based on the time instant when the signal arrives at the AP.
	By multiplying the received signal $\mathbf{Y}_{q}$ in~\eqref{channel_estimation} with the normalized pilot signal $\boldsymbol{\phi}_{k}/\sqrt{\tau_{p}}$, we get
    \begingroup\makeatletter\def\f@size{10}\check@mathfonts
		\begin{equation}	\label{estimation}
			\mathbf{y}_{q t_{k}} = \sum\nolimits_{i = 1}^{K} \sqrt{\frac{p_{i}}{\tau_{p}}} (\alpha_{iqk} \mathbf{h}_{qi}^{\mathfrak{B}_{r}} ) + \mathbf{n}_{q t_{k}},	 
		\end{equation} 
	\endgroup
	where $t_{k} \in \left\{ 1,2,\cdots ,\tau_p \right\}$ is the index of the pilot assigned to user $k$, $\alpha_{i q k} = \boldsymbol{\phi}_{i}^{T} \mathbf{\Theta}_{qi}^{T} \boldsymbol{\phi}_{k}^{*}$, $\mathbf{n}_{qt_{k}}=\frac{1}{\sqrt{\tau_{p}}} \mathbf{N}_{q} \boldsymbol{\phi}_{k}^{*} \sim \mathcal{N}_{\mathbb{C}}\left(\mathbf{0}_M, \sigma^{2} \mathbf{I}_{M}\right)$ is the resultant effective noise at the $q$-th AP. The linear minimum mean square error (LMMSE) estimate of $\mathbf{h}_{q k}^{\mathfrak{B}_{r}}$ as follows:
	\begingroup\makeatletter\def\f@size{10}\check@mathfonts
		\begin{equation} \label{estimation_result}
			\hat{\mathbf{h}}_{q k}^{\mathfrak{B}_{r}} = \sqrt{\frac{p_k}{\tau_{p}}} \left( \alpha_{kqk} \right)^{*} \mathbf{R}_{qk} \left(\mathbf{\Psi}_{qt_{k}} \right)^{-1} \mathbf{y}_{q t_{k}}, 
		\end{equation}
	\endgroup
	where $\mathbf{\Psi}_{q t_{k}} = \sum\nolimits_{i = 1}^{K} \frac{p_{i}}{\tau_{p}} \left|\alpha_{iqk}\right|^{2} + \sigma^{2} \mathbf{I}_{M}$ is the covariance matrix of the received signal $\mathbf{y}_{q t_{k}}$.
	The estimated channel and its corresponding estimate error are uncorrelated and distributed as
	\begingroup\makeatletter\def\f@size{10}\check@mathfonts
		\begin{equation} \label{estimation_error}
			\hat{\mathbf{h}}_{q k}^{\mathfrak{B}_{r}} \sim \mathcal{N}_{\mathbb{C}}\left(\mathbf{0}_{M},\mathbf{B}_{qk}\right), \tilde{\mathbf{h}}_{q k}^{\mathfrak{B}_{r}} \sim \mathcal{N}_{\mathbb{C}}\left(\mathbf{0}_{M},\mathbf{C}_{qk} \right), 
		\end{equation}
	\endgroup
	where $\mathbf{B}_{q k} = \frac{p_k}{\tau_{p}} \left| \alpha_{k q k} \right|^{2} \mathbf{R}_{q k} \left(\mathbf{\Psi}_{q k} \right)^{-1} \mathbf{R}_{q k}$ and $\mathbf{C}_{q k}=\mathbf{R}_{q k}-\mathbf{B}_{q k}$. Besides, the respective channel estimate of user $i$ in $\mathcal{P}_k\setminus\left\{ k \right\}$ ($\mathcal{P}_k$ is the subset of users that use the same pilot sequence with user $k$) is linearly dependent with $\hat{\mathbf{h}}_{q k}^{\mathfrak{B}_{b}^{r}}$. If $\mathbf{R}_{qk}$ is invertible, we have
	\begingroup\makeatletter\def\f@size{10}\check@mathfonts
		\begin{equation}
			\hat{\mathbf{h}}_{q i}^{\mathfrak{B}_{b}^{r}}=\mathbf{E}_{i q k} \hat{\mathbf{h}}_{q k}^{\mathfrak{B}_{b}^{r}}, 
		\end{equation}
	\endgroup
	where $\mathbf{E}_{i q k}=\sqrt{\frac{p_{i}}{p_{k}}} (\frac{\alpha_{i q k}}{\alpha_{k q k}})^{*}\mathbf{R}_{q i} \mathbf{R}_{q k}^{-1}\in \mathbb{C}^{M \times M}$. 

	\section{Mixed Coherent and Non-Coherent Transmissions Under Asynchronous Reception}
    
	In this section, we first obtain the achievable sum-rate for the mixed coherent and non-coherent transmissions under asynchronous reception and then characterize the sum-rate in closed form for the case with the MR precoding scheme. Subsequently, we propose an AP clustering algorithm to group APs into a set of coherent clusters.
    
	\subsection{Mixed Coherent and Non-Coherent Transmissions}

	In the CF mMIMO-OFDM system, the quantized phase shifts introduced by asynchronous reception degrade the communication performance~\cite{10032129, 10319677}. In this letter, we consider a mixed coherent and non-coherent transmission scheme in which the APs that serve users are grouped into a set of coherent clusters. Specifically, for any user $k$, the APs in set $\mathcal{Q}_{k}$ are divided into $L_{k}$ coherent clusters. Denote $\mathcal{Q}_{k}^{\ell_{k}}$ as the set of APs in coherent cluster $\ell_{k} \in \left\{ 1, 2, \cdots, L_{k} \right\}$, and have $\sum_{\ell_{k}=1}^{L_{k}} | \mathcal{Q}_{k}^{\ell_{k}} | = | \mathcal{Q}_{k} |$. In the coherent cluster $\ell_{k}$, the APs have the same quantized phase shift $\chi_{k, n}^{\ell_{k}}$ ($\chi_{q k, n} = \chi_{k, n}^{\ell_{k}}$ for $\forall q \in \mathcal{Q}_{k}^{\ell_{k}}$) and sent the same data symbol $s_{k, n}^{\ell_{k}}$ ($\mathbb{E}\{ s_{k, n}^{\ell_{k}} (s_{k, n}^{\ell_{k}})^{H} \} = 1$) to the user $k$. Hence, transmissions are performed coherently. On the other hand, the APs in different coherent clusters have different quantized phase shifts. 
	As such, the transmissions from different clusters are performed in a non-coherent fashion, and the data symbols are independent and different, i.e. $\mathbb{E} \{ s_{k, n}^{\ell_{k}} (s_{i, n}^{\ell_{i}})^{H} \} = 0$ for $\forall k \ne i$ or ${\ell_{k}} \ne {\ell_{i}}$.At this time, the transmitted signal from AP $q$ is
	\begingroup\makeatletter\def\f@size{10}\check@mathfonts
  		\begin{equation}  \label{signal}
    		\mathbf{x}_{q} = \sum\nolimits_{i \in \mathcal{K}_{q}} \mathbf{w}_{q i} s_{i, n}^{\ell_{i}},
  		\end{equation}
	\endgroup
	where $\mathbf{w}_{q i} = \sqrt{\rho_{q i}} \mathbf{\bar{w}}_{q i} / \sqrt{\mathbb{E}\{ \| \mathbf{\bar{w}}_{q i} \|^{2} \}}$ is the precoding vector.
    Herein, $\rho_{q i} \ge 0$ is the transmit power that AP $q$ assigns to user $i$, and $\mathbf{\bar{w}}_{q i} $ is an arbitrarily scaled vector of the precoding vector. 
	In the mixed coherent and non-coherent transmissions, the received signal at user $k$ is given by 
	\begingroup\makeatletter\def\f@size{10}\check@mathfonts
		\begin{equation} \label{receive_signal}
			\begin{aligned}
				&\hat{s}_{k, n}^{\mathrm{mixed}} = \sum\nolimits_{q = 1}^{Q} (\chi_{q k, n} \mathbf{h}_{q k}^{\mathfrak{B}_{r}})^{H} \mathbf{x}_{q} + n_{k, n}  \\
				&= \sum\nolimits_{i = 1}^{K} \sum\nolimits_{q \in \mathcal{Q}_{i}} (\chi_{q k, n} \mathbf{h}_{q k}^{\mathfrak{B}_{r}})^{H} \mathbf{w}_{q i} s_{i, n}^{\ell_{i}} + n_{k, n} \\
				&= \sum\nolimits_{q \in \mathcal{Q}_{k}^{\ell_{k}}} (\chi_{k, n}^{\ell_{k}} \mathbf{h}_{q k}^{\mathfrak{B}_{r}})^{H} \mathbf{w}_{q k} s_{k, n}^{\ell_{k}} \\
				& + \sum\nolimits_{\ell_{k}' \neq \ell_{k}} \sum\nolimits_{q' \in \mathcal{Q}_{k}^{\ell_{k}'}} (\chi_{k, n}^{\ell_{k}'} \mathbf{h}_{q' k}^{\mathfrak{B}_{r}})^{H} \mathbf{w}_{q' k} s_{k, n}^{\ell_{k}'} + n_{k, n}\\
				&  + \sum\nolimits_{i \neq k} \sum\nolimits_{\ell_{i} = 1}^{L_{i}} \sum\nolimits_{q'' \in \mathcal{Q}_{i}^{\ell_{i}}} (\chi_{q'' k, n} \mathbf{h}_{q'' k}^{\mathfrak{B}_{r}})^{H} \mathbf{w}_{q'' i} s_{i, n}^{\ell_{i}} ,
			\end{aligned}
		\end{equation}
	\endgroup
	where the first term is the desired signal from cluster $\ell_{k}$, the second and third terms are the inter-cluster and inter-user interference, respectively.
	\begin{figure*}[ht] 
		\centering
		\begingroup\makeatletter\def\f@size{10}\check@mathfonts
			\begin{equation} \label{SINR}
				\begin{aligned}
					\gamma_{k}^{\mathrm{mixed}} = \frac{\sum\nolimits_{\ell_{k}=1}^{L_{k}} \left| \sum_{q \in \mathcal{Q}_{k}^{\ell_{k}}} \mathbb{E}\left\{ (\mathbf{h}_{q k}^{\mathfrak{B}_{r}})^{H} \mathbf{w}_{q k} \right\} \right|^2}{\sum\nolimits_{i = 1}^{K} \sum\nolimits_{\ell_{i} = 1}^{L_{i}} \mathbb{E}\left\{ \left| \sum_{q \in \mathcal{Q}_{i}^{\ell_{i}}} (\chi_{q k, n} \mathbf{h}_{q k}^{\mathfrak{B}_{r}})^{H} \mathbf{w}_{q i} \right|^2\right\} - \sum\nolimits_{\ell_{k}=1}^{L_{k}} \left| \sum_{q \in \mathcal{Q}_{k}^{\ell_{k}}} \mathbb{E}\left\{ (\mathbf{h}_{q k}^{\mathfrak{B}_{r}})^{H} \mathbf{w}_{q k} \right\} \right|^2 + \sigma_{\mathrm{dl}}^{2}} .
				\end{aligned}
			\end{equation}
		\endgroup
	\end{figure*}
	\begin{figure*}[ht] 
		\centering
		\begingroup\makeatletter\def\f@size{10}\check@mathfonts
			\begin{equation} \label{closed_form}
				\begin{aligned}
				\gamma_{k}^{\mathrm{mixed}} = \frac{\sum\nolimits_{\ell_{k}=1}^{L_{k}} \left| \sum_{q \in \mathcal{Q}_{k}^{\ell_{k}}} \sqrt{\rho_{qk} \mathrm{tr}(\mathbf{B}_{qk})} \right|^2}{\sum\nolimits_{i = 1}^{K} \sum\nolimits_{\ell_{i} = 1}^{L_{i}} \sum_{q \in \mathcal{Q}_{i}^{\ell_{i}}} \rho_{qi} \frac{\mathrm{tr}(\mathbf{R}_{q k} \mathbf{B}_{q i})}{\mathrm{tr}(\mathbf{B}_{q i})} + \sum\nolimits_{i \in \mathcal{P}_{k} \backslash \{k\}} \sum\nolimits_{\ell_{i} = 1}^{L_{i}} \left| \sum_{q \in \mathcal{Q}_{i}^{\ell_{i}}} \chi_{q k, n}^{\mathcal{Q}_{\ell_{k}}} \sqrt{\frac{\rho_{qi}}{\mathrm{tr}(\mathbf{B}_{q i})}} \mathrm{tr}\left( \mathbf{E}_{iqk} \mathbf{B}_{qk} \right) \right|^{2} + \sigma_{\mathrm{dl}}^{2}} .
				\end{aligned}
			\end{equation}
		\endgroup
		\hrulefill
	\end{figure*}
	
	\begin{proposition} \label{Proposition1}
		Based on the received signal in~\eqref{receive_signal}, when user $k$ uses successive interference cancellation (SIC) to detect the signals sent by $L_{k}$ coherent clusters (with the perfect SIC and arbitrary decoding order), the achievable downlink sum-rate under asynchronous reception is given by
		\begingroup\makeatletter\def\f@size{10}\check@mathfonts
			\begin{equation} \label{sum_SE}
				\mathrm{SE}_{n}^{\mathrm{mixed}} = \sum\nolimits_{k=1}^{K} \frac{\tau_{d}}{\tau} \frac{N}{N_{\mathrm{OFDM}}}\log _{2}\left( 1 + \gamma_{k}^{\mathrm{mixed}} \right), 
			\end{equation}
		\endgroup
		where the effective SINR is given in~\eqref{SINR}, as shown on top of this page.
	\end{proposition}
	
	\begin{proof}
		Please refer to Appendix A.
	\end{proof}
    
	\begin{algorithm}[t]
		\caption{AP clustering scheme based on distance criteria}
		\label{Algorithm:AP_Clustering}
		\begin{algorithmic}[1] 
		\REQUIRE $M_{0}$, $D$, $\{d_{q k}\}$, $k = 1, \cdots, K, q = 1, \cdots, Q $; 
		\STATE User $k^{*}$ chooses a group of $M_{0}$ serving APs that correspond to the $M_{0}$ nearest APs to user $k^{*}$, and get the set $\mathcal{Q}_{k^{*}}$;
		\STATE Finding the AP closest to and farthest from user $k^{*}$ in set $\mathcal{Q}_{k^{*}}$, denoted the indexes as $\hat{q}_{k^{*}}$ and $\check{q}_{k^{*}}$, and then calculating $L_{k^{*}}$ according to~\eqref{Num_cluster};
		\FOR{$q \in \mathcal{Q}_{k^{*}}$}
		  \FOR{$\ell_{k^{*}}=1$ \textbf{to} $L_{k^{*}}$}
		    \STATE $d^{\ell_{k^{*}}-1} = (\ell_{k^{*}}-1) \cdot D$, $ d^{\ell_{k^{*}}} = \ell_{k^{*}} \cdot D$
			\IF{$d^{\ell_{k^{*}}-1} \le \lfloor d_{q, k^{*}} - d_{\hat{q}_{k^{*}}, k^{*}} \rfloor < d^{\ell_{k^{*}}}$}
			  \STATE $\mathcal{Q}_{k^{*}}^{\ell_{k^{*}}} \Leftarrow \mathcal{Q}_{k^{*}}^{\ell_{k^{*}}} \cup \{ q \}$;
			\ENDIF
		  \ENDFOR
		\ENDFOR
		\ENSURE $\{ \mathcal{Q}_{k^{*}}^{\ell_{k^{*}}} \}, \ell_{k^{*}}=1, \cdots, L_{k^{*}}$.
		\end{algorithmic}
	\end{algorithm}

	From the effective SINR in~\eqref{SINR}, it is observed that the desired signal in the numerator contains the sum of the squared contributions from different coherent clusters, which is different from the synchronous reception case where the summation is inside the square and the asynchronous case where the summation contains the different quantized phase shifts.
	By clustering, APs with the same quantized phase shift are grouped into the same coherent cluster, so the impact of the quantized phase shift caused by asynchronous reception on the desired signal is eliminated. Furthermore, the multi-user interference in the denominator of~\eqref{SINR} also changes accordingly, but the effect of the quantized phase shifts still exists. Therefore, the mixed coherent and non-coherent transmissions can effectively reduce the impact of asynchronous reception on CF mMIMO-OFDM systems.

	\begin{corollary} \label{Corollary 1}
		For the mixed coherent and non-coherent transmissions under asynchronous reception, if $\mathbf{\bar{w}}_{q k} = \mathbf{\hat{h}}_{q k}^{\mathfrak{B}_{r}}$ (i.e., MR precoding) is used, the SINR given in~\eqref{SINR} can be explicitly rewritten as in~\eqref{closed_form}, as shown on top of previous page.
	\end{corollary}
    \begin{proof}
		This proof follows from~\cite{10319677}, and the details are omitted for brevity.
	\end{proof}


	\subsection{AP Clustering Algorithms}

	In this subsection, we propose an efficient AP clustering algorithm based on distance criteria. Firstly, the users are connected to a fixed number of APs, which aims to ensure service quality~\cite{10279001}. 
	Then, the connected APs are further grouped, where the APs with the same quantized phase shift are divided into the same coherent cluster. 
	 
	Specifically, when a new user $k^{*}$ enters the network, it chooses a group of $M_{0}$ serving APs that correspond to the $M_{0}$ nearest APs to user $k^{*}$. 
	The reason is that selecting APs based on distance can reduce the number of coherent clusters and improve the performance of mixed transmissions~\cite{10279001}. The selected AP set is $\mathcal{Q}_{k^{*}}$. 
	Then, finding the AP closest to and farthest from user $k^{*}$ in set $\mathcal{Q}_{k^{*}}$, denoted the indexes as $\hat{q}_{k^{*}}$ and $\check{q}_{k^{*}}$ respectively.
    Subsequently, the number of coherent clusters of user $k^{*}$ is determined, which is given by 
	\begingroup\makeatletter\def\f@size{10}\check@mathfonts
		\begin{equation} \label{Num_cluster}
			L_{k^{*}} = \left \lfloor (d_{\check{q}_{k^{*}} k^{*}} - d_{\hat{q}_{k^{*}} k^{*}}) / D \right \rfloor + 1 = \delta_{\check{q}_{k^{*}} k^{*}} + 1. 
		\end{equation}
	\endgroup
	According to the number $L_{k^{*}}$ of coherent clusters of user $k^{*}$ and the sampling distance $D$, setting a set of reference distances $\{d^{\ell_{k^{*}}} = \ell_{k^{*}} \cdot D: \ell_{k^{*}} = 0, 1, \cdots,L_{k^{*}}\}$.
	If AP $q$ in set $\mathcal{Q}_{k^{*}}$ satisfying $d^{\ell_{k^{*}}-1} \le | d_{q, k^{*}} - d_{\hat{q}_{k^{*}}, k^{*}} | < d^{\ell_{k^{*}}}$, then AP $q$ is grouped into the coherent cluster $\ell_{k^{*}}$, i.e., $q \in \mathcal{Q}_{k^{*}}^{\ell_{k^{*}}}$. The AP clustering algorithm is summarized in \textbf{Algorithm~\ref{Algorithm:AP_Clustering}}. 
	

	\begin{figure}
		\centering
		\includegraphics[width=0.8\linewidth]{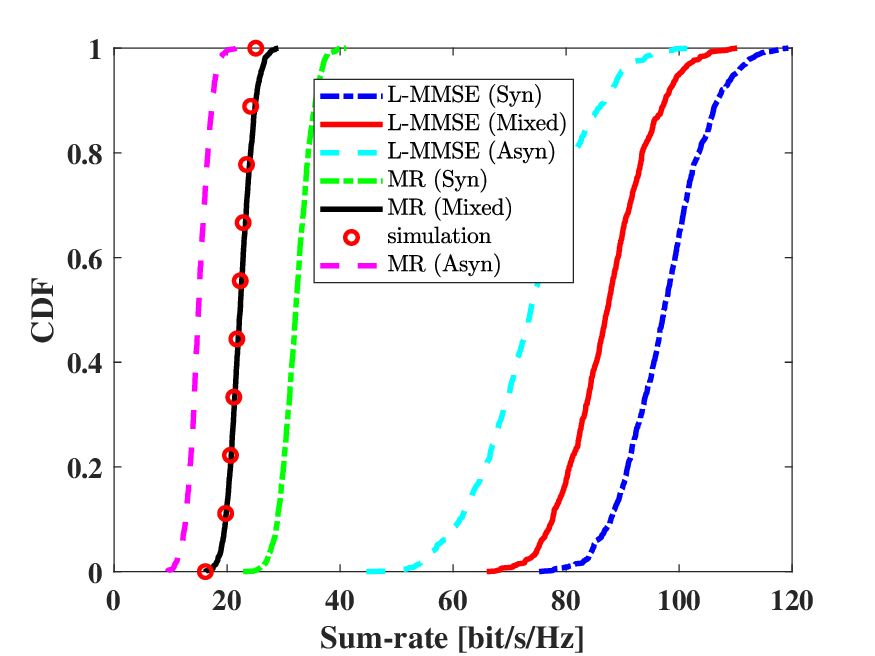}
		\caption{CDF of sum-rate comparing the synchronous reception, asynchronous reception, and mixed transmissions. }
		\label{Sum_rate}
	\end{figure}
	
	\section{Numerical Results}
    
	In this section, numerical results are provided to verify our theoretical analysis. Consider a CF mMIMO-OFDM system setup where $Q = 30$ APs and $K = 20$ users are independently and uniformly distributed in a $1 \times 1$ km square area and a wrap-around topology is used. MR and local minimum mean square error (L-MMSE) precoding are used at downlink in all simulations. 
	The channel model and the OFDM parameters are the same as in~\cite{10319677}. 
	The coherence time and bandwidth are set as $0.5 \: \text{ms}$ and $180 \: \text{kHz}$, respectively, which fits a CB setup of $N_{\mathrm{sub}} = 14$ subcarriers and $N_{\mathrm{T}} = 7$ OFDM symbols. The uplink transmit power is $100 \: \text{mW}$ and the total transmit power of each AP is $200 \: \text{mW}$. 
	For the convenience of comparison with the clustering algorithm in~\cite{10279001}, we use the same power allocation scheme as in~\cite{10279001}, namely equal power allocation.
	Each CB contains 98 channel uses, in which 84 channel uses is used for downlink data transmission~\cite{10319677}.

	\begin{figure}
		\centering
		\includegraphics[width=0.8\linewidth]{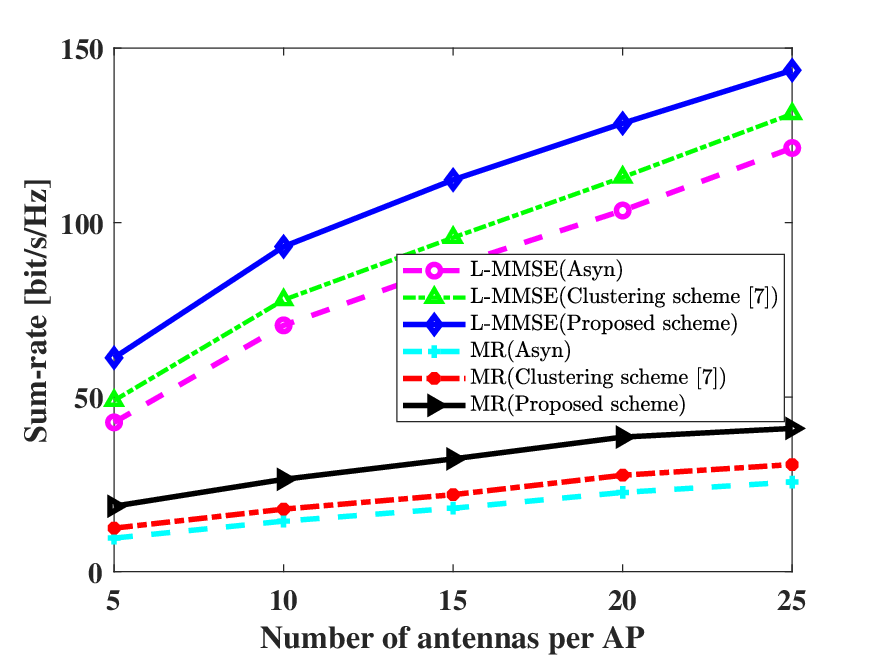}
		\caption{Sum-rate of users versus the number of antennas per AP.}
		\label{AP_change}
	\end{figure}

	Fig.~\ref{Sum_rate} illustrates the cumulative distribution function (CDF) of sum-rate under synchronous reception, asynchronous reception, and mixed coherent and non-coherent transmission scenarios using L-MMSE and MR precoding under the setup $\{ M, M_{0} \} = \{ 10, 20 \}$. As shown in Fig.~\ref{Sum_rate}, for both the cases with the L-MMSE or MR precoding, the proposed mixed coherent and non-coherent transmission method always achieves a higher sum-rate than that under asynchronous reception, which is lower than that under synchronous reception. Compared with asynchronous reception, the mixed transmissions improve system sum-rate by coherently transmitting signals within the coherent cluster. However, due to the non-coherent transmission between coherent clusters, this performance improvement is not enough to achieve the ideal synchronous reception performance, which can be inferred from \textbf{Proposition~\ref{Proposition1}}. Although there is still a certain gap between the sum-rate of mixed transmissions and synchronous reception, it greatly alleviates the performance loss caused by asynchronous reception to CF mMIMO-OFDM systems.


	In Fig.~\ref{AP_change}, we compare the performance of the proposed clustering algorithm with the Fixed algorithm in~\cite{10279001}. For a fair comparison, the number of selected APs is $M_{0} = 20$ and the number of CPUs for the Fixed algorithm is 1. Note that the Fixed algorithm proposed in~\cite{10279001} only selects $M_{0}$ largest large-scale fading coefficients of APs to serve users, which cannot group the APs with the same quantized into the same coherent clusters. Therefore, the Fixed algorithm is the user-centric approach and cannot adopt the mixed transmissions scheme.
	From Fig.~\ref{AP_change}, it can be seen that the Fixed algorithm (the user-centric approach) has a limited positive effect on asynchronous CF mMIMO-OFDM systems because there are still different quantized phase shifts between APs~\cite{10032129}.  In contrast, the proposed clustering algorithm can group APs with the same quantized phase shift into the same coherent cluster. After applying the mixed transmissions, the quantization phase shift effect on the desired signal can be eliminated. Therefore, the performance of asynchronous CF mMIMO-OFDM systems with mixed transmissions can be significantly improved.


	\section{Conclutions}
	
	In this letter, we analyzed the performance gain of the mixed coherent and non-coherent transmission approach in enhancing the sum-rate of CF mMIMO-OFDM systems under asynchronous reception. First, we derived a mixed coherent and non-coherent transmission achievable downlink sum-rate expression under asynchronous reception, and then obtained a closed-form expression when using the MR precoding. Subsequently, we proposed an effective AP clustering algorithm that can effectively group APs into a set of coherent clusters to use the mixed coherent and non-coherent transmissions. 
	The results showed that the mixed coherent and non-coherent transmission can effectively improve the sum-rate of CF mMIMO-OFDM systems under asynchronous reception, while still suffering mile performance loss compared to synchronous reception. Numerical results verified the effectiveness of the proposed method.
	

	\appendices

	\section{Proff of Proposition 1}

	Under the assumption of perfect SIC, at the beginning of data detection, user $k$ is unaware of any transmission signals. Without loss of generality, data detection starts from the coherent cluster 1. As such, the received signal in~\eqref{receive_signal} can be written as
	\begingroup\makeatletter\def\f@size{9}\check@mathfonts
		\begin{equation}  \label{first_decode}
			\begin{aligned}
				&\hat{s}_{k, n}^{\mathrm{mixed}, 1} = (\chi_{k, n}^{1})^{*} \sum\nolimits_{q \in \mathcal{Q}_{k}^{1}} 
				\mathbb{E} \left\{ (\mathbf{h}_{q k}^{\mathfrak{B}_{r}})^{H} \mathbf{w}_{q k} \right\} s_{k, n}^{1} \\
				&+ (\chi_{k, n}^{1})^{*} \sum\nolimits_{q \in \mathcal{Q}_{k}^{1}} \left((\mathbf{h}_{q k}^{\mathfrak{B}_{r}})^{H} \mathbf{w}_{q k} - \mathbb{E} \left\{ (\mathbf{h}_{q k}^{\mathfrak{B}_{r}})^{H} \mathbf{w}_{q k} \right\} \right) s_{k, n}^{1} \\
				&+ \sum\nolimits_{\ell_{k}' = 2} (\chi_{k, n}^{\ell_{k}'})^{*} \sum\nolimits_{q' \in \mathcal{Q}_{k}^{\ell_{k}'}} ( \mathbf{h}_{q' k}^{\mathfrak{B}_{r}})^{H} \mathbf{w}_{q' k} s_{k, n}^{\ell_{k}'} \\
				&+ \sum\nolimits_{i \neq k} \sum\nolimits_{\ell_{i} = 1}^{L_{i}} \sum\nolimits_{q'' \in \mathcal{Q}_{i}^{\ell_{i}}} (\chi_{q'' k, n} \mathbf{h}_{q'' k}^{\mathfrak{B}_{r}})^{H} \mathbf{w}_{q'' i} s_{i, n}^{\ell_{i}} + n_{k, n},
			\end{aligned}
		\end{equation}
	\endgroup
	where the first term is the desired signal over known deterministic channel while other terms are treated as uncorrelated noise~\cite{8809413}. Sequentially, user $k$ detects signal from coherent cluster $\ell_{k}$ by subtracting the first $\ell_{k}-1$ signals:
	\begingroup\makeatletter\def\f@size{9}\check@mathfonts
		\begin{equation}   \label{last_decode}
			\begin{aligned} 
				&\hat{s}_{k, n}^{\mathrm{mixed}, \ell_{k}} = \hat{s}_{k, n}^{\mathrm{mixed}} - \sum\nolimits_{t=1}^{\ell_{k}-1} (\chi_{k, n}^{t})^{*} \sum_{q \in \mathcal{Q}_{k}^{t}} \mathbb{E} \left\{ (\mathbf{h}_{q k}^{\mathfrak{B}_{r}})^{H} \mathbf{w}_{q k} \right\} s_{k, n}^{t} \\
				&= (\chi_{k, n}^{\ell_{k}})^{*} \sum\nolimits_{q \in \mathcal{Q}_{k}^{\ell_{k}}} \mathbb{E} \left\{ (\mathbf{h}_{q k}^{\mathfrak{B}_{r}})^{H} \mathbf{w}_{q k} \right\} s_{k, n}^{\ell_{k}} \\
				&+ (\chi_{k, n}^{\ell_{k}})^{*} \sum\nolimits_{q \in \mathcal{Q}_{k}^{\ell_{k}}} \left( (\mathbf{h}_{q k}^{\mathfrak{B}_{r}})^{H} \mathbf{w}_{q k} - \mathbb{E} \left\{ (\mathbf{h}_{q k}^{\mathfrak{B}_{r}})^{H} \mathbf{w}_{q k} \right\} \right) s_{k, n}^{\ell_{k}} \\
				&+ \sum_{t=1}^{\ell_{k}-1} (\chi_{k, n}^{t})^{*} \sum\nolimits_{q \in \mathcal{Q}_{k}^{t}} \left( (\mathbf{h}_{q k}^{\mathfrak{B}_{r}})^{H} \mathbf{w}_{q k} - \mathbb{E} \left\{ (\mathbf{h}_{q k}^{\mathfrak{B}_{r}})^{H} \mathbf{w}_{q k} \right\} \right) s_{k, n}^{t} \\
				&+ \sum\nolimits_{\ell_{k}' = \ell_{k}+1}^{L_{k}} (\chi_{k, n}^{\ell_{k}'})^{*} \sum\nolimits_{q' \in \mathcal{Q}_{k}^{\ell_{k}'}} (\mathbf{h}_{q' k}^{\mathfrak{B}_{r}})^{H} \mathbf{w}_{q' k} s_{k, n}^{\ell_{k}'} \\
				&+ \sum\nolimits_{i \neq k} \sum\nolimits_{\ell_{i} = 1}^{L_{i}} \sum\nolimits_{q'' \in \mathcal{Q}_{i}^{\ell_{i}}} (\chi_{q'' k, n} \mathbf{h}_{q'' k}^{\mathfrak{B}_{r}})^{H} \mathbf{w}_{q'' i} s_{i, n}^{\ell_{i}} + n_{k, n}.
		  \end{aligned}
		\end{equation}
	\endgroup
	\begin{figure*}[ht] 
		\centering
		\begingroup\makeatletter\def\f@size{9}\check@mathfonts
			\begin{equation} \label{SINR_part}
				\begin{aligned}
					\gamma_{\ell_{k}, k}^{\mathrm{mixed}} = \frac{\left| \sum_{q \in \mathcal{Q}_{k}^{\ell_{k}}} \mathbb{E}\left\{ (\mathbf{h}_{q k}^{\mathfrak{B}_{r}})^{H} \mathbf{w}_{q k} \right\} \right|^2}{\sum\nolimits_{i = 1}^{K} \sum\nolimits_{\ell_{i} = 1}^{L_{i}} \mathbb{E}\left\{ \left| \sum_{q \in \mathcal{Q}_{i}^{\ell_{i}}} (\chi_{q k, n} \mathbf{h}_{q k}^{\mathfrak{B}_{r}})^{H} \mathbf{w}_{q i} \right|^2\right\} - \sum\nolimits_{t = 1}^{\ell_{k}} \left| \sum\nolimits_{q \in \mathcal{Q}_{k}^{t}} \mathbb{E}\left\{ (\mathbf{h}_{q k}^{\mathfrak{B}_{r}})^{H} \mathbf{w}_{q k} \right\} \right|^2 + \sigma_{\mathrm{dl}}^{2}} .
				\end{aligned}
			\end{equation}
		\endgroup
		\hrulefill
	\end{figure*}
    We treat the sum of the last four terms in~\eqref{last_decode} as uncorrelated noise, and its power $v_{\ell_{k},k}$ is
    \begingroup\makeatletter\def\f@size{8}\check@mathfonts
		\begin{equation}   \label{uncorrelated_noise}
		  \begin{aligned} 
			&\mathbb{E}\left\{ \left| v_{\ell_{k},k} \right|^2 \right\} = \mathbb{E} \left\{ \left| \sum\nolimits_{q \in \mathcal{Q}_{k}^{\ell_{k}}} \left( (\mathbf{h}_{q k}^{\mathfrak{B}_{r}})^{H} \mathbf{w}_{q k} - \mathbb{E} \left\{ (\mathbf{h}_{q k}^{\mathfrak{B}_{r}})^{H} \mathbf{w}_{q k} \right\} \right) \right|^2  \right\}	\\
			& \quad + \sum_{t=1}^{\ell_{k}-1} \mathbb{E} \left\{ \left| \sum\nolimits_{q \in \mathcal{Q}_{k}^{t}} \left( (\mathbf{h}_{q k}^{\mathfrak{B}_{r}})^{H} \mathbf{w}_{q k} - \mathbb{E} \left\{ (\mathbf{h}_{q k}^{\mathfrak{B}_{r}})^{H} \mathbf{w}_{q k} \right\} \right) \right|^2 \right\} \\
			& \quad + \sum\nolimits_{\ell_{k}' = \ell_{k}+1}^{L_{k}} \mathbb{E} \left\{ \left| \sum\nolimits_{q' \in \mathcal{Q}_{k}^{\ell_{k}'}} (\mathbf{h}_{q' k}^{\mathfrak{B}_{r}})^{H} \mathbf{w}_{q' k} \right|^2 \right\} + \sigma_{\mathrm{dl}}^{2}\\
			& \quad + \sum\nolimits_{i \neq k} \sum\nolimits_{\ell_{i} = 1}^{L_{i}} \mathbb{E} \left\{ \left| \sum\nolimits_{q'' \in \mathcal{Q}_{i}^{\ell_{i}}} (\chi_{q'' k, n} \mathbf{h}_{q'' k}^{\mathfrak{B}_{r}})^{H} \mathbf{w}_{q'' i} \right|^2 \right\} \\
			& = \sum\nolimits_{i = k}^{K} \sum\nolimits_{\ell_{i} = 1}^{L_{i}} \mathbb{E} \left\{ \left| \sum\nolimits_{q \in \mathcal{Q}_{i}^{\ell_{i}}} (\chi_{q k, n} \mathbf{h}_{q k}^{\mathfrak{B}_{r}})^{H} \mathbf{w}_{q i} \right|^2 \right\} \\
			& \quad - \sum\nolimits_{t = 1}^{\ell_{k} - 1} \left| \sum\nolimits_{q' \in \mathcal{Q}_{k}^{t}} \mathbb{E} \left\{ (\mathbf{h}_{q k}^{\mathfrak{B}_{r}})^{H} \mathbf{w}_{q k} \right\} \right|^2 + \sigma_{\mathrm{dl}}^{2}.
		  \end{aligned}
		\end{equation}
	\endgroup
	Then, we have $\mathrm{SE}_{\ell_{k}, k,n}^{\mathrm{mixed}} = \frac{\tau-\tau_{p}}{\tau} \frac{N}{N_{\mathrm{OFDM}}} \log_2(1+\gamma_{\ell_{k}, k}^{\mathrm{mixed}})$, where SINR $\gamma_{\ell_{k}, k}^{\mathrm{mixed}}$ is given in~\eqref{SINR_part}, shown at the bottom of this page.
    Similar to non-coherent transmission, the spectral efficiency of user $k$ $\mathrm{SE}_{k,n}^{\mathrm{mixed}} = \sum_{\ell_{k} = 1}^{L_{K}} \mathrm{SE}_{\ell_{k}, k,n}^{\mathrm{mixed}}$ is given by 
	\begingroup\makeatletter\def\f@size{9}\check@mathfonts
		\begin{equation} \label{SINR_sum}
			\begin{aligned}
				\mathrm{SE}_{k,n}^{\mathrm{mixed}} = \frac{\tau_{d}}{\tau} \frac{N}{N_{\mathrm{OFDM}}} \log_2\left( \prod_{\ell_{k} = 1}^{L_{K}} (1+\gamma_{\ell_{k}, k}^{\mathrm{mixed}})\right).
			\end{aligned}
		\end{equation}
	\endgroup
    By substituting~\eqref{SINR_part} into $\prod_{\ell_{k} = 1}^{L_{K}} (1+\gamma_{\ell_{k}, k}^{\mathrm{mixed}})$ of \eqref{SINR_sum}, we final obtain $\gamma_{k}^{\mathrm{mixed}}$.

	\bibliographystyle{IEEEtran}
	\bibliography{IEEEabrv,mybib}

\begin{thebibliography}{1}
\providecommand{\url}[1]{#1}
\csname url@samestyle\endcsname
\providecommand{\newblock}{\relax}
\providecommand{\bibinfo}[2]{#2}
\providecommand{\BIBentrySTDinterwordspacing}{\spaceskip=0pt\relax}
\providecommand{\BIBentryALTinterwordstretchfactor}{4}
\providecommand{\BIBentryALTinterwordspacing}{\spaceskip=\fontdimen2\font plus
\BIBentryALTinterwordstretchfactor\fontdimen3\font minus \fontdimen4\font\relax}
\providecommand{\BIBforeignlanguage}[2]{{%
\expandafter\ifx\csname l@#1\endcsname\relax
\typeout{** WARNING: IEEEtran.bst: No hyphenation pattern has been}%
\typeout{** loaded for the language `#1'. Using the pattern for}%
\typeout{** the default language instead.}%
\else
\language=\csname l@#1\endcsname
\fi
#2}}
\providecommand{\BIBdecl}{\relax}
\BIBdecl

\bibitem{10422885}
J.~Zheng, J.~Zhang, H.~Du, D.~Niyato, B.~Ai, M.~Debbah, and K.~B. Letaief, ``Mobile cell-free massive {MIMO}: Challenges, solutions, and future directions,'' \emph{IEEE Wireless Commun.}, vol.~31, no.~3, pp. 140--147, Jun. 2024.

\bibitem{7827017}
H.~Q. Ngo, A.~Ashikhmin, H.~Yang, E.~G. Larsson, and T.~L. Marzetta, ``Cell-free massive {MIMO} versus small cells,'' \emph{IEEE Trans. Wireless Commun.}, vol.~16, no.~3, pp. 1834--1850, Mar. 2017.

\bibitem{8809413}
{\"O}.~{\"O}zdogan, E.~Bj{\"o}rnson, and J.~Zhang, ``Performance of cell-free massive {MIMO} with rician fading and phase shifts,'' \emph{IEEE Trans. Wireless Commun.}, vol.~18, no.~11, pp. 5299--5315, Nov. 2019.

\bibitem{10684238}
M.~Mohammadi, Z.~Mobini, H.~Q. Ngo, and M.~Matthaiou, ``Next-generation multiple access with cell-free massive {MIMO},'' \emph{Proc. IEEE}, pp. 1--49, 2024.

\bibitem{10032129}
J.~Zheng, J.~Zhang, J.~Cheng, V.~C.~M. Leung, D.~W.~K. Ng, and B.~Ai, ``Asynchronous cell-free massive {MIMO} with rate-splitting,'' \emph{IEEE J. Sel. Areas Commun.}, vol.~41, no.~5, pp. 1366--1382, May. 2023.

\bibitem{10319677}
G.~Li, S.~Wu, C.~You, W.~Zhang, G.~Shang, and X.~Zhou, ``Cell-free massive {MIMO-OFDM}: Asynchronous reception and performance analysis,'' \emph{IEEE Internet Things J.}, vol.~11, no.~7, pp. 11\,894--11\,906, Apr. 2024.

\bibitem{10279001}
R.~P. Antonioli, I.~M. Braga, G.~Fodor, Y.~C. Silva, and W.~C. Freitas, ``Mixed coherent and non-coherent transmission for multi-cpu cell-free systems,'' in \emph{Proc. IEEE Int. Conf. Commun. (ICC)}, Rome, Italy, May. 2023, pp. 1068--1073.

\bibitem{10459246}
U.~Kunnath~Ganesan, T.~T. Vu, and E.~G. Larsson, ``Cell-free massive {MIMO} with multi-antenna users and phase misalignments: A novel partially coherent transmission framework,'' \emph{IEEE Open J. Commun. Soc.}, vol.~5, pp. 1639--1655, Apr. 2024.

\bibitem{8341954}
F.~Yang, P.~Cai, H.~Qian, and X.~Luo, ``Pilot contamination in massive {MIMO} induced by timing and frequency errors,'' \emph{IEEE Trans. Wireless Commun.}, vol.~17, no.~7, pp. 4477--4492, Jul. 2018.

\end{thebibliography}
	
\end{document}